\documentclass[twoside,bezier,12pt, reqno]{amsart}
\usepackage{amssymb,latexsym,epsfig,mathrsfs,graphpap,graphics,amsmath,amssymb}
\usepackage{amstext}
\usepackage{amsthm}
\usepackage[utf8]{inputenc}
\usepackage[english]{babel}
\usepackage{helvet}
\usepackage{courier}

\newtheorem{theorem}{Theorem}[section]
\newtheorem{lemma}{Lemma}[section]
\newtheorem{proposition}{Proposition}[section]

\theoremstyle{Definition}
\newtheorem{definition}{Definition}[section]

\theoremstyle{remark}
\newtheorem{remark}[theorem]{Remark}
\numberwithin{equation}{section}

\begin{document}

\begin{flushleft}
 {\bf\Large { Quadratic-phase wave packet transform}}

\parindent=0mm \vspace{.2in}

{\bf{M. Younus Bhat$^{1},$  Aamir H. Dar$^{2}$, Didar Urynbassarova $^{3,4}$ and Altyn Urynbassarova$^{5, 6}$}}
\end{flushleft}

{{\it $^{1}$ Department of  Mathematical Sciences,  Islamic University of Science and Technology Awantipora, Pulwama, Jammu and Kashmir 192122, India. E-mail: $\text{ gyounusg@gmail.com}$}}

{{\it $^{2}$ Department of  Mathematical Sciences,  Islamic University of Science and Technology Awantipora, Pulwama, Jammu and Kashmir 192122, India. E-mail: $\text{ahdkul740@gmail.com}$}}
{{\it $^{3}$ School of Information and Electronics, Beijing Institute of Technology, Beijing 100081, China}}
{{\it $^{4}$  Beijing Key Laboratory of Fractional Signals and Systems, Beijing 100081, China
{{\it $^{5}$   Department of Technology and Ecology, School of Society, Technology and Ecology, Narxoz University, 
Almaty 050000, Kazakhstan

{{\it $^{6}$  Department of Information Security, Faculty of Information Technology, Eurasian National University, 
Nur-Sultan 010000, Kazakhstan

}}

\begin{quotation}
\noindent
{\footnotesize {\sc Abstract.} The quadratic phase Fourier transform(QPFT) has gained much popularity in recent years because of its applications in
image and signal processing. However, the QPFT is
inadequate for localizing the quadratic-phase spectrum which is
required in some applications. In this paper, the quadratic-phase wave packet transform (QP-WPT) is
proposed to address this problem, based on the wave packet transform (WPT) and QPFT. Firstly, we propose the  definition
of the QP-WPT and gave its  relation with windowed Fourier transform (WFT). Secondly, several notable inequalities and important properties of newly defined QP-WPT, such as  boundedness, reconstruction formula, Moyal's formula , Reproducing kernel are derived.
Finally, we formulate several classes of uncertainty inequalities such as Leib's uncertainty principle, logarithmic uncertainty inequality and  the Heisenberg uncertainty inequality.\\


}
\end{quotation}
\section{ \bf Introduction}
\noindent

 The Fourier transform (FT) is an important tool in optical communication and signal processing \cite{x1}.However, owing to its global kernel the FT is
incapable of obtaining information about local properties of the signal. But, the actual signals are often  non-stationary or time-variable, so to overcome this problem, the short time Fourier transform (STFT) is employed  that uses a
time window of fixed length applied at regular intervals so that we can obtain a portion of the signal considered to be stationary \cite{x2}. The resulting time-varying spectral depiction  is critical for non-stationary signal analysis, but in
this case it comes at fixed spectral and temporal resolution. The wavelet analysis(WT) \cite{x3,x4}  provides an attractive and pinch-hitting tool to the STFT by using an optical multichannel correlator with a bank of WT filters, can provide a better illustration of the signal instead of the STFT. Nonetheless, in the high  frequency region WT has poor frequency resolution. 
To solve this defect the wave packet transform (WPT) was proposed by combining the merits of STFT and WT \cite{x5,x6}. WPT is a linear transform which  uses the Weyl operator and the wave packages.

In recent years, researchers have  successfully applied wave packet transform (WPT)  in the fields of  wireless communication, denoising,  and image compression \cite{x7}-\cite{x14}. Wave packet transform (WPT) is used widely in signal processing as it has some better morality than wavelet transform (WT) \cite{x15,x16}. Moreover, it can realize multilevel decomposition and analyze the high frequency decomposition that is not achieved in traditional discrete  WT. The frequency subbands of signal are selected via wave packet decomposition,
that improves the time-frequency resolution capability of the signal. However, the WPT is defined as the FT of the signal 
windowed with the wavelet, so the results obtained by WPT will not be optimal in dealing with chirp signals whose energy is not
well concentrated in FT domain.

 A superlative generalized version  of the Fourier transform(FT) called  quadratic-phase Fourier transform(QPFT)  has been introduced by  Castro et al\cite{4a,5a}. This novel transform  has overthrown all the the applicable signal processing tools as it provides a unified analysis of both  transient
and non-transient signals in an easy and insightful fashion. The  QPFT is actually  a generalization of several well known transforms like Fourier, fractional Fourier and linear canonical transforms whose kernel is in the exponential form. Due to its extra degrees of freedom, the quadratic phase
Fourier transform (QFT) has marked its importance in treatment of
problems demanding several controllable parameters arising in diverse branches of science and
engineering, including harmonic analysis, sampling, image
processing, and so on \cite{a}-\cite{xo2}.

Recently Prasad and  Sharma \cite{ak} introduced the quadratic phase Fourier wavelet transform (QPFWT),which is generalization of classical continuous
wavelet transform,\cite{x18}-\cite{x24}continuous fractional wavelet transform \cite{x25,x26,x27} as well as generalization of linear canonical
wavelet transform\cite{x27,x28}. QPFWT intertwine the advantages of the quadratic-phase Fourier and wavelet transforms into a novel integral transform which
assimilates their individual properties. 
However, the transform neither relies on the complete kernel of the QPFT nor exhibits any existing
convolution structure in the quadratic-phase Fourier transform. So Shah and Lone \cite{fs1} introduced Quadratic-phase wavelet in different approach which is completely reliant upon convolution associated with QPFT.

As one of the generalization of the classical WPT, the fractional WPT(Fr-WPT) and linear canonical wave packet transform (LC-WPT)  have been introduced to improve the performance in concentration
\cite{x34,x36,x37}. They have attained a much more attention of the  signal processing community and optics. But to the  best of our knowledge theory about quadratic phase wave packet transform (QP-WPT) have never been proposed up to date, therefore it is worthwhile to study the
theory of QP-WPT based on the wave packet transform (WPT) and QPFT which can be productive for signal processing theory and applications. Therefore, the cynosure of this paper  is to rigorously study the quadratic phase wave packet transform (QP-WPT).\\

The highlights of the paper are pointed out below:
\begin{itemize}
\item To introduce a novel integral transform coined as the quadratic phase wave packet transform.\\\

\item To establish  relationship between  quadratic phase wave packet transform with   Fourier transform(FT) and windowed Fourier transform (WFT).\\\

\item To study several notable inequalities and important properties of newly defined QP-WPT, such as  boundedness, reconstruction formula, Moyal's formula , Reproducing kernel.\\\
\item To formulate several classes of uncertainty inequalities, such as  Leib-type, the logarithmic uncertainty inequalities  and the Heisenberg-type uncertainty inequalities associated with
the quadratic phase wave packet transform.\\\

\end{itemize}

The paper is organised as follows. In Section \ref{sec 2}, we provide some preliminary results required in subsequent sections. In Section \ref{sec 3}, we provide  the  definition of  quadratic phase wave packet transform (QP-WPT). Then we investigated  several basic properties of the QP-WPT  which are important for signal representation in signal processing. In Section \ref{sec 4}, we develop a series of uncertainty inequalities such as  Leib's uncertainty principle, the logarithmic uncertainty inequality and the Heisenberg-type inequality associated with the quadratic phase
 wave packet transform. Finally, a conclusion is extracted in Section \ref{sec 5}.

\section{Preliminaries}\label{sec 2}
In this section we recall some basic concepts and notations, which will be useful in our
study on quadratic-phase wave packet transfrom(QPWPT).
\subsection{Fourier transfrom}
We use the following definition of Fourier transform [1] on $L^1(\mathbb R)$ space
\begin{equation}\label{ft}
\mathcal F[f](\xi)=\frac{1}{\sqrt{2\pi}}\int_{\mathbb R}f(t)e^{-i t\xi}dt, \quad \forall \xi\in\mathbb R.
\end{equation}

\subsection{Continuous wavelet transform}
Wavelet transform presents an attractive alternative to the STFT by using a time-frequency
window that changes with frequency, which can effectively provide resolution of varying granularity.
The continuous wavelet
transform (CWT) of a signal $f(t)\in L^2(\mathbb R)$ is defined as \cite{x20,x21} 
\begin{equation}\label{cwt}
CWT_f(\beta,\alpha)=\frac{1}{\sqrt{\alpha}}\int_{\mathbb R}f(t)\psi^{*}\left(\frac{t-\beta}{\alpha}\right)dt,
\end{equation}
where $*$ denotes the complex conjugate and $t$ is time, $\beta$ is the translation parameter, $\alpha$ is the scaling parameter and  $\psi(t)$ is the
transforming function, called mother wavelet. Here $\alpha >0$  and $\psi$  is normalized such that the$\|\psi\|=1$ in $L^2(\mathbb R)$ space.
\subsection{Windowed Fourier transform}
The windowed Fourier transform of $f(t)\in L^2(\mathbb R)$ with respect to the windowed function $\phi\in L^2(\mathbb R)$ is defined as \cite{x38}
\begin{equation}\label{wft}
\mathcal G_\phi[f](w,\beta)=\int_{\mathbb R}f(t)\phi^*(t-\beta)e^{-i\xi t}dt
\end{equation}
and the inverse of the function $f(t)\in L^2(\mathbb R)$ is defined by [24]
\begin{equation}\label{iwft}
f(t)=\frac{b}{2\pi\langle\phi,\psi\rangle}\int_{\mathbb R}\int_{\mathbb R}\mathcal G_\phi[f](\xi,\beta)e^{i\xi t}\psi(t-\beta)d\xi d\beta,
\end{equation}
where, $\psi\in L^2(\mathbb R).$
\subsection{Wave packet transform}
The wave packet transform (WPT) combines elements of STFT and CWT, and it can be viewed as \cite{x7,x8}
\begin{equation}\label{wpt}
WPT_f(\xi,\beta,\alpha)=\frac{1}{\sqrt{2\pi\alpha}}\int_{\mathbb R}f(t)\overline{\psi_{\alpha}\left({t-\beta}\right)}e^{-i\xi t}dt
\end{equation}
where $\psi_{\alpha}\left({t-\beta}\right)=\psi\left(\frac{t-\beta}{\alpha}\right).$

The WPT is the Fourier transform of a signal windowed with a wavelet that is dilated by $\alpha$ and translated by $\beta$.

\begin{lemma}\label{lem11}\cite{x36}
Let $\psi\in L^p(\mathbb R),\quad p\in [1,\infty).$ Then, $\|\psi_\alpha(t-\beta)\|_{L^p(\mathbb R)}=\alpha^{(1/p-1/2)}\|\psi\|_{L^p(\mathbb R)},$
\end{lemma}
\subsection{Quadratic-phase   Fourier transform }\ \\
 In this subsection we introduce the Quadratic-phase Fourier transform  which is a neoteric addition to the classical integral transforms and we also gave its inversion formula and some other classical results which are already present in literature.

\begin{definition}\label{qpft} Given a  parameter $\mu=(a,b,c,d,e),$  the QPFT of any signal $f$ is defined by \cite{ak}
\begin{equation}\label{eqnqpft}
\mathcal Q_\mu[f](\xi)=\int f(t)K_{\mu}(t,\xi)dt,
\end{equation}
where $K_{\mu}(t,\xi)$ is  the quadratic-phase Fourier kernel, given by
\begin{equation}\label{eqnker}
 K_{\mu}(t,\xi) =\sqrt{\frac{b}{2\pi i}}e^{(at^2+bt\xi+c\xi+dt+e\xi)}
\end{equation}

with $a,b,c,d,e \in\mathbb R,\quad b\ne 0.$
\end{definition}
\begin{theorem}\label{thinvqpft}The inversion formula of the quadratic-phase Fourier transform is given by \cite{ak}
\begin{equation}\label{eqninvqpft}
f(t)=\int\mathcal Q_{\mu}[f](\xi)\overline{K_{\mu}(t,\xi)}d\xi.
\end{equation}
\end{theorem}

Using the inversion theorem, we can get the Parseval's relation given by \cite{ak}
\begin{equation}
\langle f,g \rangle=\langle\mathcal Q_{\mu}[f],\mathcal Q_{\mu}[g] \rangle
\end{equation}
and Plancherel identity is given by
\begin{equation}
\int|\mathcal Q_{\mu}[f](\xi)|^2d\xi=\int|f(t)|^2dt.
\end{equation}
\begin{theorem}\label{nth2}\cite{fs1,fs2}Let $f,g\in L^2(\mathbb R)$ and $\alpha,\beta,\tau\in \mathbb R$ then
\begin{itemize}
\item$\mathcal Q_{\mu}[\alpha f+\beta g](\xi)=\alpha \mathbb Q_{\mu}[f](w)+\beta\mathbb Q_{\mu}[g](w).$
\item$\mathcal Q_{\mu}[f(t-\tau)](\xi)=exp\{-i(a\tau^2+b\tau \xi+d\tau)\}\mathcal Q_{\mu}[e^{-2ia\tau t}f(t)](\xi).$
\item$\mathcal Q_{\mu}[f(-t)](\xi)=\mathcal Q_{\mu'}[f(t)](-\xi),\quad \mu'=(a,b,c,-d,-e).$
\item$\mathcal Q_{\mu}[e^{i\alpha t}f(t)](\xi)=exp\{i(\alpha^2+2\alpha b \xi+\alpha eb)\frac{1}{b}\}\mathcal Q_{\mu}[f]\left(w+\frac{a}{b}\right).$
 \item$\mathcal Q_{\mu}[\overline{f(t)}](w)=\overline{\mathcal Q_{-\mu}[{f(t)}](w)}.$
\end{itemize}
\end{theorem}
\begin{theorem}[Convolution\cite{ak}]\label{nth3}If $f,g\in L^2(\mathbf R)$ then
\begin{equation}
\mathcal Q_{\mu}[f\ast_{\mu} g](\xi)=\sqrt{\frac{2\pi i}{b}}e^{-i(c\xi^2+e\xi)}\mathcal Q_{\mu}[f](\xi)\mathcal Q_{\mu}[e^{-ia(\cdot)^2-id(\cdot)}g](\xi).
\end{equation}
Where \begin{equation}\label{convdef}
(f\ast_{\mu} g)(t)=\int_{\mathbf R}f(x)g(t-x)e^{-ia(t^2-z^2)-id(t-z)}dz.
\end{equation}
\end{theorem}

\subsection{\bf Quadratic-phase wavelet transform }\ \\
The generalization of the classical continuous
wavelet transform, continuous fractional wavelet transform, as well as generalization of linear canonical
wavelet transform, is the quadratic phase wavelet transform (QPWT).

  For a signal  $f(t)\in\L^2(\mathbb R^2),$  the continuous quadratic-phase wavelet transform of $f$ with respect to an
analyzing wavelet $\psi\in L^2(\mathbb R)$ and the parameter set $\mu = (a, b,c,d,e)$  is defined by \cite{ak}
\begin{equation}\label{cqpwt}
CQPWT_f(\beta,\alpha)=\sqrt{\frac{b}{2\pi  i}}\int_{\mathbb R}f(t)\overline{\psi^\mu_{\beta,\alpha}(t)}dt,
\end{equation}
where the family $\psi^\mu_{\beta,\alpha}(t)$ is called quadratic-phase wavelet (QPW) and is given by
 \begin{equation}
 \psi^\mu_{\beta,\alpha}(t)=\frac{1}{\sqrt{\alpha}}\psi\left(\frac{t-\beta}{\alpha}\right)e^{-ia(t^2-\beta^2)-id(t-\beta)}.
 \end{equation}
\begin{lemma}\label{psi}\cite{ak}
If $\psi\in L^2(\mathbb R),$ Then $\psi^\mu_{\beta,\alpha}\in L^2(\mathbb R)$ with $\|\psi^\mu_{\beta,\alpha}\|^2=\|\psi\|^2.$
\end{lemma}
Now we  are ready to introduce a novel integral transform the quadratic phase wave packet transform.

\section{\bf Quadratic-phase wavelet packet transform(QP-WPT)}\ \\
\label{sec 3}
In this section, we propose a  definition of QP-WPT based on the idea of WPT by adding extra dimension to kernel and wavelet. We replace the wave packet transform kernel by the QPFT kernel and the wavelet by the
quadratic-phase wavelet (QPW).

\begin{definition}[QP-WPT]\label{def qpwpt}
The QP-WPT transform of a function $f\in\L^2(\mathbb R)$ with respect to wavelet function $\psi$ is defined as
\begin{eqnarray}\label{eqn qpwpt}
\nonumber W^\mu_f(\xi,\beta,\alpha)&=&\int_{\mathbb R}f(t)\overline{\psi^\mu_{\beta,\alpha}(t)}K_\mu(t,\xi)dt\\
\nonumber &=&\sqrt{\frac{b}{2\pi i}}\int_{\mathbb R}e^{i(at^2+bt\xi+c\xi^2+dt+e\xi)}f(t)\overline{\psi^\mu_{\beta,\alpha}(t)}dt,\\
\end{eqnarray}
where $\psi^\mu_{\beta,\alpha}(t)=\psi_{\alpha}\left(t-\beta\right)e^{-ia(t^2-\beta^2)-id(t-\beta)}$ and $\psi_\alpha(t)=\frac{1}{\alpha}\psi\left(\frac{t}{\alpha}\right).$
\end{definition}
\begin{remark} By varying the parameter $\mu=(a,b,c,d,e)$
Definition \ref{def qpwpt} embodies certain existing time-frequency transforms and also give birth to some novel
time-frequency tools which are yet to be reported in the open literature which are listed below:
\begin{itemize}
\item For $\mu=(a/2b,-1/b,c/2b,0,0)$, Definition \ref{def qpwpt} boils down to the novel linear canonical wave packet
transform\begin{equation*}
W^\mu_f(\xi,\beta,\alpha)=\int_{\mathbb R}f(t)K_\mu(t,\xi)\overline{\psi\left(\frac{t-\beta}{\alpha}\right)}e^{i\frac{a}{2b}(t^2-\beta^2)}dt
\end{equation*}
\item For $\mu=(\cot\theta,-\csc\theta,\cot\theta,0,0)$, $\theta\ne n\pi$
Definition \ref{def qpwpt} reduces to the novel fractional
wave packet
transform\begin{equation*}
W^\mu_f(\xi,\beta,\alpha)=\int_{\mathbb R}f(t)K_\mu(t,\xi)\overline{\psi\left(\frac{t-\beta}{\alpha}\right)}e^{i\cot\theta(t^2-\beta^2)}dt
\end{equation*}
\item For $\mu=(1,b,0,1,0)$, $b\ne 0$, we can obtain the novel Fresnel wave packet
transform\begin{equation*}
W^\mu_f(\xi,\beta,\alpha)=\int_{\mathbb R}f(t)K_\mu(t,\xi)\overline{\psi\left(\frac{t-\beta}{\alpha}\right)}e^{i(t^2-\beta^2)+id(t-\beta)}dt
\end{equation*}
\item For $\mu=(0,-1,1,0,0)$,
Definition \ref{def qpwpt} reduces to the classical
wave packet
transform\begin{equation*}
W^\mu_f(\xi,\beta,\alpha)=\int_{\mathbb R}f(t)K_\mu(t,\xi)\overline{\psi\left(\frac{t-\beta}{\alpha}\right)}dt
\end{equation*}
\end{itemize}
\end{remark}
\begin{theorem}\label{th1} Let $W^\mu_f(\xi,\beta,\alpha)$ and $Q_\mu[f]$ be the QP-WPT and QPFT of a function $f\in L^2(\mathbb R),$ respectively and  let $\psi^\mu_{\beta,\alpha}$ be the QPW, then we have
\begin{eqnarray}
\nonumber W^\mu_f(\xi,\beta,\alpha)&=&\sqrt{\alpha}\int_{\mathbb R}K_\mu(w,\beta)e^{-i[c(\alpha w)^2+e(\alpha w)-2cw\xi]}\\\
\nonumber&&\qquad\qquad\mathcal Q_\mu[e^{ia(.)^2+id(.)}f(t)](w+\xi)Q_\mu[e^{-ia(.)^2-id(.)}\psi(.)](\alpha w)dw\\
\end{eqnarray}
\end{theorem}
\begin{proof}
Let us denote
\begin{equation*}f_{\xi,\mu}=\sqrt{\frac{b}{2\pi i}}e^{i(at^2+bt\xi+c\xi^2+dt+e\xi)}f(t).\end{equation*}
On taking QPFT on both sides of above equation, we have
\begin{eqnarray*}
\mathcal Q_\mu[f_{\xi,\mu}]&&\\\
&&=\int_{\mathbb R}K(w,t)f_{\xi,\mu}(t)dt\\\
&&=\int_{\mathbb R}\sqrt{\frac{b}{2\pi i}}e^{i(at^2+btw+cw^2+dt+ew)}\sqrt{\frac{b}{2\pi i}}e^{i(at^2+bt\xi+c\xi^2+dt+e\xi)}f(t)dt\\\
&&=\sqrt{\frac{b}{2\pi i}}\int_{\mathbb R}\sqrt{\frac{b}{2\pi i}}e^{i[at^2+bt(w+\xi)+c(w+\xi)^2+dt+e(w+\xi)]}\\\
&&\qquad\qquad\qquad\qquad\times e^{i(at^2+dt)} f(t)e^{-i(2cw\xi)}dt\\\
&&=\sqrt{\frac{b}{2\pi i}}e^{-i2cw\xi}\mathcal Q_\mu[e^{i(at^2+dt)}f(t)](w+\xi).
\end{eqnarray*}
From \cite{ak}, we have
\begin{eqnarray*}
Q_\mu[\psi^\mu_{\beta,\alpha}](w)&&\\
&&=\sqrt{\alpha}e^{i(a\beta^2+b\beta w+cw^2+d\beta+ew)-ic(\alpha w)^2-ie(\alpha w)}\\
&&\qquad\qquad\times Q_\mu[e^{-ia(.)^2-id(.)}\psi(.)](\alpha w).
\end{eqnarray*}
The QP-WPT is represented in terms of inner product of $f_{\xi,\mu}$ and $\psi^\mu_{\beta,\alpha}$ and by Parseval theorem of QPFT, we have
\begin{eqnarray*}
W^\mu_f(\xi,\beta,\alpha)&&\\\
&&=\langle f_{\xi,\mu},\psi^\mu_{\beta,\alpha}\rangle\\\
&&=\left\langle \mathcal Q_\mu[f_{\xi,\mu}], Q_\mu[\psi^\mu_{\beta,\alpha}]\right\rangle\\\
&&=\sqrt{\frac{\alpha b}{2\pi i}}\int_{\mathbb R}e^{i(a\beta^2+b\beta w+cw^2+d\beta+ew-c(\alpha w)^2-e(\alpha w)-2cw\xi)}\\\
&&\qquad\qquad\qquad\times\mathcal Q_\mu[e^{i(at^2+dt)}f(t)](w+\xi)Q_\mu[e^{-ia(.)^2-id(.)}\psi(.)](\alpha w)dw.
\end{eqnarray*}
Now using \ref{eqnker}, we get the desired proof.
\end{proof}

Further, the definition of the QP-WPT in \ref{eqn qpwpt} can be rewritten as

\begin{eqnarray}\label{eqn2 qpwpt}
\nonumber W^\mu_f(\xi,\beta,\alpha)&&\\\
\nonumber&&=\sqrt{\frac{b}{2\pi i}}\int_{\mathbb R}e^{i(at^2+bt\xi+c\xi^2+dt+e\xi)+ia(t^2-\beta^2)+id(t-\beta)}\\\
\nonumber&&\qquad\qquad\qquad\times f(t)\overline{\psi_\alpha(t-\beta)}dt\\\
\nonumber&&=\int_{\mathbb R}f(t)\overline{\psi^\mu_{\xi,\beta,\alpha,}}dt,\\\
\end{eqnarray}
where

\begin{equation}
\psi^\mu_{\xi,\beta,\alpha,}(t)=\overline{\left(\sqrt{\frac{b}{2\pi i}}\right)}e^{-i(at^2+bt\xi+c\xi^2+dt+e\xi)-ia(t^2-\beta^2)-id(t-\beta)}\psi_\alpha(t-\beta).
\end{equation}

\begin{proposition}[Relation with WFT]
\begin{eqnarray}\label{rel wft}
\nonumber W^\mu_f(\xi,\beta,\alpha)&&\\\
\nonumber&&=\sqrt{\frac{b}{2\pi i}}\int_{\mathbb R}e^{i(at^2+bt\xi+c\xi^2+dt+e\xi)+ia(t^2-\beta^2)+id(t-\beta)}f(t)\overline{\psi_\alpha(t-\beta)}dt\\\
\nonumber&&=e^{i(c\xi^2+e\xi-a\beta^2-d\beta)}\int_{\mathbb R}\sqrt{\frac{b}{2\pi i}}e^{i(2at^2+bt\xi+2dt)}f(t)\overline{\psi_\alpha(t-\beta)}dt\\\
\nonumber&&=e^{i(c\xi^2+e\xi-a\beta^2-d\beta)}\int_{\mathbb R}\sqrt{\frac{b}{2\pi i}}e^{i(2at^2+2dt)}f(t)\overline{\psi_\alpha(t-\beta)}e^{ibt\xi}dt\\\
\nonumber&&=e^{i(c\xi^2+e\xi-a\beta^2-d\beta)}\mathcal G_{\psi^\mu}[h](b\xi ,\beta)\\\
\end{eqnarray}
\end{proposition}
where $h(t)=\sqrt{\frac{b}{2\pi i}}e^{i(2at^2+2dt)}f(t)$

\subsection{Basic properties of the QP-WPT}
In this subsection we prove some notable inequalities associated with the QP-WPT. Moreover, we also investigate some basic properties of the QP-WPT which are important for signal representation in signal processing.  
\begin{lemma}\label{lem bdd}
Let $\psi_\alpha\in\L^p(\mathbb R)$ and $f\in L^q(\mathbb R)$ and $p,q\in [1,\infty)$  with $\frac{1}{p}+\frac{1}{q},$ then
\begin{equation}\label{eqn bdd}
|W^\mu_f(\xi,\beta,\alpha)|\le \alpha^{1/p-1/2}\sqrt{\frac{b}{2\pi}}\|\psi\|_{L^p(\mathbb R)}\|f\|_{L^q(\mathbb R)}.
\end{equation}
\end{lemma}
\begin{proof}
 From  (\ref{eqn2 qpwpt}), we have
 \begin{eqnarray*}
|W^\mu_f(\xi,\beta,\alpha)|&=&\left|\sqrt{\frac{b}{2\pi i}}\int_{\mathbb R}e^{i(at^2+bt\xi+c\xi^2+dt+e\xi)+ia(t^2-\beta^2)+id(t-\beta)}f(t)\overline{\psi_\alpha(t-\beta)}dt\right|\\\
&=&\sqrt{\frac{b}{2\pi }}\left|\int_{\mathbb R}e^{i(at^2+bt\xi+c\xi^2+dt+e\xi)+ia(t^2-\beta^2)+id(t-\beta)}f(t)\overline{\psi_\alpha(t-\beta)}dt\right|\\\
&&\le\sqrt{\frac{b}{2\pi }}\left|\int_{\mathbb R}f(t)\overline{\psi_\alpha(t-\beta)}dt\right|.\\\
\end{eqnarray*}
By Lemma \ref{lem11} and Holder's inequality, above yields
\begin{equation*}
|W^\mu_f(\xi,\beta,\alpha)|\le \alpha^{1/p-1/2}\sqrt{\frac{b}{2\pi}}\|\psi\|_{L^p(\mathbb R)}\|f\|_{L^q(\mathbb R)}
\end{equation*}
which completes the proof.
\end{proof}
\begin{theorem}[Boundedness] For $\psi,f\in\L^2(\mathbb R),$  the QP-WPT is bounded on $\L^2(\mathbb R)$.
\end{theorem}
\begin{proof}
 By taking $p=q=2$ in  Lemma \ref{lem bdd},we have:
\begin{equation*}
|W^\mu_f(\xi,\beta,\alpha)|\le \sqrt{\frac{b}{2\pi}}\|\psi\|_{L^2(\mathbb R)}\|f\|_{L^2(\mathbb R)}
\end{equation*}
which shows that the QP-WPT is bounded on $L^2(\mathbb R)$.
\end{proof}
\begin{theorem} Let $\psi\in L^p(\mathbb R)$ and $f\in L^1(\mathbb R)\cap L^1(\mathbb R).$ Then we have
\begin{equation}
\|W^\mu_f(\xi,\beta,\alpha)\|_{L^p(\mathbb R)}\le \alpha^{(1/p-1/2)}\sqrt{\frac{b}{2\pi }}\|\psi\|_{L^p(\mathbb R)}\|f\|_{L^1(\mathbb R)}.
\end{equation}
\end{theorem}
\begin{proof}
By applying the Minkowski’s  inequality to (\ref{eqn2 qpwpt}), we obtain
\begin{eqnarray*}
\|W^\mu_f(\xi,\beta,\alpha)\|_{L^p(\mathbb R)}&&\\\
&&=\left(\int_{\mathbb R}\left|\sqrt{\frac{b}{2\pi i}}\int_{\mathbb R}e^{i(at^2+bt\xi+c\xi^2+dt+e\xi)+ia(t^2-\beta^2)+id(t-\beta)}f(t)\overline{\psi_\alpha(t-\beta)}dt\right|^pd\beta\right)^{1/p}\\\
&&\le\sqrt{\frac{b}{2\pi }}\int_{\mathbb R}\left(\int_{\mathbb R}\left|f(t)\overline{\psi_\alpha(t-\beta)}\right|^pd\beta\right)^{1/p}dt.
\end{eqnarray*}
Setting $t-\beta=y,$ we have

\begin{eqnarray*}
\|W^\mu_f(\xi,\beta,\alpha)\|_{L^p(\mathbb R)}&=&\sqrt{\frac{b}{2\pi }}\int_{\mathbb R}\left(\int_{\mathbb R}\left|f(t)\overline{\psi_\alpha(y)}\right|^pdy\right)^{1/p}dt\\\
&&\le\sqrt{\frac{b}{2\pi }}\int_{\mathbb R}\left(\int_{\mathbb R}\left|\overline{\psi_\alpha(y)}\right|^pdy\right)^{1/p}|f(t)|dt\\\
&&\le\sqrt{\frac{b}{2\pi }}\|\psi_\alpha\|_{L^p(\mathbb R)}\|f\|_{L^1(\mathbb R)}\\\
&&\le \alpha^{(1/p-1/2)}\sqrt{\frac{b}{2\pi }}\|\psi\|_{L^p(\mathbb R)}\|f\|_{L^1(\mathbb R)}.
\end{eqnarray*}
Which completes the proof.

\end{proof}

\begin{theorem}[Reconstruction theorem]\label{th recon} Every signal $f\in L^2(\mathbb R),$ can be reconstructed from QP-WPT by the formula
\begin{eqnarray}\label{recon eqn}
f(t)=\int_{\mathbb R}\int_{\mathbb R}W^\mu_f(\xi,\beta,\alpha)\psi^\mu_{\xi,\beta,\alpha}(t)d\xi d\beta.\\\
\end{eqnarray}
\end{theorem}
\begin{proof}
Let $h(t),\psi,\phi\in L^2(\mathbb R).$ Assuming $\langle\phi,\psi\rangle\ne0$ and $\psi_\alpha$ as a windowed function then by the
inverse of the WFT (\ref{iwft}), we have
\begin{eqnarray*}
h(t)=\frac{b}{2\pi \langle\phi,\psi\rangle}\int_{\mathbb R}\int_{\mathbb R}\mathcal G_{\psi^\mu}[h](b\xi ,\beta)e^{-ibt\xi}\psi_\alpha(t-\beta)d\xi d\beta.
\end{eqnarray*}
By virtue of (\ref{rel wft}), we have from above equation
\begin{eqnarray}
\nonumber\sqrt{\frac{b}{2\pi i}}e^{i(2at^2+2dt)}f(t)&=&\frac{b}{2\pi \langle\phi,\psi\rangle}\int_{\mathbb R}\int_{\mathbb R}e^{-i(c\xi^2+e\xi-a\beta^2-d\beta)}e^{-ibt\xi}\\\
\nonumber&&\qquad\qquad\qquad\times\psi_\alpha(t-\beta)
W^\mu_f(\xi,\beta,\alpha)d\xi d\beta\\\
\nonumber f(t)&=&\sqrt{\frac{b i}{2\pi}}\frac{1}{\langle\phi,\psi\rangle}\int_{\mathbb R}\int_{\mathbb R}e^{-i(c\xi^2+e\xi-a\beta^2-d\beta+2dt+bt\xi+2at^2)}\\\
\nonumber&&\qquad\qquad\qquad\times\psi_\alpha(t-\beta)
W^\mu_f(\xi,\beta,\alpha)d\xi d\beta\\\
\nonumber&=&\frac{1}{\langle\phi,\psi\rangle}\int_{\mathbb R}\int_{\mathbb R}\psi^\mu_{\xi,\beta,\alpha}(t)W^\mu_f(\xi,\beta,\alpha)d\xi d\beta.\\\
\end{eqnarray}
For perfect reconstruction take $\langle\phi,\psi\rangle=1,$ above equation yields
\begin{eqnarray*}
f(t)=\int_{\mathbb R}\int_{\mathbb R}W^\mu_f(\xi,\beta,\alpha)\psi^\mu_{\xi,\beta,\alpha}(t)d\xi d\beta.\\\
\end{eqnarray*}
Which completes the proof.
\end{proof}
\begin{theorem}[Moyal's Formula]\label{thm ortho}Let $W^\mu_f(\xi,\beta,\alpha)$ and $W^\mu_g(\xi,\beta,\alpha)$ be the QP-WPT with respect to the wavelets $\psi$ and $\phi$ respectively, then
\begin{equation}
\langle W^\mu_f(\xi,\beta,\alpha),W^\mu_g(\xi,\beta,\alpha)\rangle_{L^2(\mathbb R^2)}=\overline{\langle \psi,\phi \rangle}_{L^2(\mathbb R)}\langle f,g \rangle_{L^2(\mathbb R)}
\end{equation}
\end{theorem}
\begin{eqnarray*}
\langle W^\mu_f(\xi,\beta,\alpha),W^\mu_g(\xi,\beta,\alpha)\rangle&&\\\
&&=\int_{\mathbb R^2}W^\mu_f(\xi,\beta,\alpha)\overline{W^\mu_g(\xi,\beta,\alpha)}d\xi d\beta\\\
&&=\int_{\mathbb R^2}\left\{\int_{\mathbb R}f(t)\overline{\psi^\mu_{\beta,\alpha}(t)}K_\mu(t,\xi)dt\right.\\\
&&\qquad\qquad\qquad\times\left.\int_{\mathbb R}\overline{g(t')}\phi^\mu_{\beta,\alpha}(t')\overline{K_\mu(t',\xi)}dt'\right\}d\xi d\beta\\\
&&=\int_{\mathbb R^2}\int_{\mathbb R}f(t)\overline{g(t')}\phi^\mu_{\beta,\alpha}(t')\overline{\psi^\mu_{\beta,\alpha}(t)}\int_{\mathbb R}K_\mu(t,\xi)\overline{K_\mu(t',\xi)}d\xi dtdt'd\beta\\\
&&=\int_{\mathbb R^2}\int_{\mathbb R}f(t)\overline{g(t')}\phi_{\alpha}(t'-\beta)\overline{\psi_{\alpha}(t-\beta)}\frac{b}{2\pi}\int_{\mathbb R}e^{ib\xi(t-t')}d\xi dtdt'd\beta\\\
&&=\int_{\mathbb R^2}\int_{\mathbb R}f(t)\overline{g(t')}\phi_{\alpha}(t'-\beta)\overline{\psi_{\alpha}(t-\beta)}\delta(t-t')dt'dtd\beta\\\
&&=\int_{\mathbb R}\int_{\mathbb R}f(t)\overline{g(t)}\phi_{\alpha}(t-\beta)\overline{\psi_{\alpha}(t-\beta)}dtd\beta\\\
&&=\int_{\mathbb R}f(t)\overline{g(t)}dt\int_{\mathbb R}\frac{1}{\alpha}\phi\left(\frac{t-\beta}{\alpha}\right)\overline\psi\left(\frac{t-\beta}{\alpha}\right)d\beta\\\
&&=\overline{\langle \psi,\phi \rangle}_{L^2(\mathbb R)}\langle f,g \rangle_{L^2(\mathbb R)}.
\end{eqnarray*}
 Which completes the proof.

 Consequences of the Theorem \ref{thm ortho}:
\begin{itemize}
\item If $\psi=\phi$, then \begin{equation}\langle W^\mu_f(\xi,\beta,\alpha),W^\mu_g(\xi,\beta,\alpha)\rangle_{L^2(\mathbb R^2)}=\|\psi\|^2_{L^2(\mathbb R)}\langle f,g \rangle_{L^2(\mathbb R)}.\end{equation}
  \item If $\psi=\phi$,and  $f=g$ then
   \begin{equation}\label{con 2}\langle W^\mu_f(\xi,\beta,\alpha),W^\mu_g(\xi,\beta,\alpha)\rangle_{L^2(\mathbb R^2)}=\|\psi\|^2_{L^2(\mathbb R)}\|f\|^2_{L^2(\mathbb R)}.\end{equation}
 \item If $\psi=\phi=1$,and  $f=g$ then \begin{equation}\label{energy}\langle W^\mu_f(\xi,\beta,\alpha),W^\mu_g(\xi,\beta,\alpha)\rangle_{L^2(\mathbb R^2)}=\|f\|^2_{L^2(\mathbb R)}.\end{equation}
\end{itemize}
\begin{remark}[Energy conservation] Equation (\ref{energy}) yields the conservation of energy for the QP-WPT
\begin{equation}\label{energy2}
\int_{\mathbb R^2}\left| W^\mu_f(\xi,\beta,\alpha)\right|^2d\xi d\beta=\int_{\mathbb R}\left|f(t)\right|^2 dt.
\end{equation}
\end{remark}
\begin{theorem}[Reproducing kernel]\label{thm rep ker}
Let $(\xi_0,\beta_0,\alpha)$ be any point on the plane of $(\xi,\beta,\alpha),$ the necessary and sufficient condition that the function $ W^\mu_f(\xi,\beta,\alpha)$ is the QP-WPT of some function is that $ W^\mu_f(\xi,\beta,\alpha)$ must satisfy the following reproducing kernel formula
\begin{equation}\label{eqn rep ker}
 W^\mu_f(\xi,\beta,\alpha)=\int_{\mathbb R}\int_{\mathbb R}W^\mu_f(\xi,\beta,\alpha)\mathbb K_{\psi^\mu}(\xi,\beta,\alpha:\xi_0,\beta_0,\alpha)d\xi d\beta \\
\end{equation}
where  $W^\mu_f(\xi_0,\beta_0,\alpha)$ is value of function $ W^\mu_f(\xi,\beta,\alpha)$ at $(\xi_0,\beta_0,\alpha),$ and $\mathbb K_{\psi^\mu}(\xi,\beta,\alpha:\xi_0,\beta_0,\alpha)$ is called the reproducing kernel given by
\begin{equation}\label{r p ker}
\mathbb K_{\psi^\mu}(\xi,\beta,\alpha:\xi_0,\beta_0,\alpha)=\langle \psi^\mu_{\xi,\beta,\alpha},\psi^\mu_{\xi_0,\beta_0,\alpha}\rangle
\end{equation}
\end{theorem}
\begin{proof}
From (\ref{eqn2 qpwpt})and (\ref{recon eqn}), we have
\begin{eqnarray*}
W^\mu_f(\xi,\beta,\alpha)&&\\\
&&=\int_{\mathbb R}f(t)\overline{\psi^\mu_{\xi,\beta,\alpha}}(t)dt\\\
&&=\int_{\mathbb R}\left\{\int_{\mathbb R}\int_{\mathbb R}W^\mu_f(\xi,\beta,\alpha)\psi^\mu_{\xi,\beta,\alpha}(t)d\xi d\beta \right\}\overline{\psi^\mu_{\xi,\beta,\alpha}}(t)dt.\\\
\end{eqnarray*}
Setting $(\xi,\beta,\alpha)=(\xi_0,\beta_0,\alpha),$ we have
\begin{eqnarray*}
W^\mu_f(\xi_0,\beta_0,\alpha)&&\\\
&&=\int_{\mathbb R}\left\{\int_{\mathbb R}\int_{\mathbb R}W^\mu_f(\xi,\beta,\alpha)\psi^\mu_{\xi,\beta,\alpha}(t)d\xi d\beta \right\}\overline{\psi^\mu_{\xi_0,\beta_0,\alpha}}(t)dt\\\
&&=\int_{\mathbb R}\int_{\mathbb R}\int_{\mathbb R}W^\mu_f(\xi,\beta,\alpha)\psi^\mu_{\xi,\beta,\alpha}(t)\overline{\psi^\mu_{\xi_0,\beta_0,\alpha}}(t)dtd\xi d\beta \\\
&&=\int_{\mathbb R}\int_{\mathbb R}W^\mu_f(\xi,\beta,\alpha)\left\{\int_{\mathbb R}\psi^\mu_{\xi,\beta,\alpha}(t)\overline{\psi^\mu_{\xi_0,\beta_0,\alpha}}(t)dt\right\}d\xi d\beta \\\
&&=\int_{\mathbb R}\int_{\mathbb R}W^\mu_f(\xi,\beta,\alpha)\mathbb K_{\psi^\mu}(\xi,\beta,\alpha:\xi_0,\beta_0,\alpha)d\xi d\beta \\\
\end{eqnarray*}
Which completes the proof.
\end{proof}
\section{\bf Uncertainty Principle's for the QP-WPT}\label{sec 4}
  Uncertainty principle has applications in two main areas:
 harmonic analysis and signal analysis. This principle in harmonic analysis stems from the uncertainty
principle in quantum mechanics, which tells that
a particle’s velocity and position cannot be measured with infinite precision. In signal
analysis, it tells that if one observes a signal only for a finite time, then the knowledge
about the frequencies consisted by the signal is lost.
In this section, we first prove QP-WPT Lieb’s uncertainty principle by
considering the  relationship between the WFT and QP-WPT. Then we will obtain
a logarithmic uncertainty principle associated with the QP-WPT by using the relation fundamental  between FT and QP-WPT. Finally, we wil establish
a generalization of the Heisenberg type uncertainty principle for the QP-WPT.

\begin{theorem}[Leib's uncertainty principle]For $\psi,f\in L^2(\mathbb R)$ and $2\le p <\infty,$ the following inequality holds:
\begin{eqnarray}\label{eqn leibs}
\int_{\mathbb R}\int_{\mathbb R}\left|W^\mu_f(\xi,\beta,\alpha)\right|^pd\xi d\beta&\le&\frac{2}{p}(M_\mu)^p\left(\|f\|_2\|\psi\|_2\right)^p
\end{eqnarray}
where $(M_\mu)=(2\pi)^{\frac{-1}{2}}|b|^{\frac{1}{2}-\frac{1}{p}}$
\end{theorem}
\begin{proof}
The Lieb's uncertainty principle for the windowed Fourier transform \cite{x39,x38}
reads
\begin{equation}\label{l1}
\int_{\mathbb R}\int_{\mathbb R}\left|\mathcal G_{\psi}[f](\xi,\beta)\right|^pd\xi d\beta\le\frac{2}{p}\left(\|f\|_2\|\psi\|_2\right)^p
\end{equation}
for all $f,\psi\in L^2(\mathbb R)$ and $2\le p <\infty.$

For $f\in L^2(\mathbb R)$ we have function $h(t)=\sqrt{\frac{b}{2\pi i}}e^{i(2at^2+2dt)}f(t)\in L^2(\mathbb R)$ , therefore we can replace $f$ in (\ref{l1})by $h$ as:
\begin{eqnarray}
\nonumber\int_{\mathbb R}\int_{\mathbb R}\left|\mathcal G_{\psi^\mu_{\alpha}}[h](\xi,\beta)\right|^pd\xi d\beta &\le&\frac{2}{p}\left(\|h\|_2\|\psi^\mu_{\alpha}\|_2\right)^p\\
\label{l2}&=&\frac{2}{p}\left(\left(\int_{\mathbb R}\left|\sqrt{\frac{b}{2\pi i}}e^{i(2at^2+2dt)}f(t)\right|^2dt \right)^{\frac{1}{2}}\|\psi_{\alpha}\|_2\right)^p.
\end{eqnarray}
Substituting $\xi=b\xi$ in (\ref{l2}), we have
\begin{equation}\label{l3}
\int_{\mathbb R}\int_{\mathbb R}|b|\left|\mathcal G_{\psi^\mu_{\alpha}}[h](b\xi,\beta)\right|^pd\xi d\beta\le\frac{2}{p}\left(\frac{b}{2\pi}\right)^{\frac{p}{2}}\left(\left(\int_{\mathbb R}\left|f(t)\right|^2dt \right)^{\frac{1}{2}}\|\psi_{\alpha}\|_2\right)^p.
\end{equation}
Using (\ref{rel wft}) in (\ref{l3})
\begin{equation}\label{l4}
\int_{\mathbb R}\int_{\mathbb R}\left|e^{-i(c\xi^2+e\xi-a\beta^2-d\beta)}W^\mu_f(\xi,\beta,\alpha)\right|^pd\xi d\beta\le\frac{2}{p|b|}\left(\frac{b}{2\pi}\right)^{\frac{p}{2}}\left(\left(\int_{\mathbb R}\left|f(t)\right|^2dt \right)^{\frac{1}{2}}\|\psi_{\alpha}\|_2\right)^p.
\end{equation}
On further simplifying (\ref{l4}) and  using lemma \ref{psi}, we have
\begin{eqnarray*}\label{l5}
\int_{\mathbb R}\int_{\mathbb R}\left|W^\mu_f(\xi,\beta,\alpha)\right|^pd\xi d\beta&\le&\frac{2}{p|b|}\left(\frac{b}{2\pi}\right)^{\frac{p}{2}}\left(\|f\|_2\|\psi\|_2\right)^p\\\
&=&\frac{2}{p}\left(\frac{1}{|b|^{\frac{1}{p}}}\right)^p\left(\frac{|b|^{\frac{1}{2}}}{(2\pi)^{\frac{1}{2}}}\right)^p\left(\|f\|_2\|\psi\|_2\right)^p
\end{eqnarray*}
which completes the proof.
\end{proof}
\begin{lemma}[Relation between QP-WPT and FT]\label{rel ft}
We have from (\ref{eqn2 qpwpt})
\begin{eqnarray}
\nonumber W^\mu_f(\xi,\beta,\alpha)&&\\\
\nonumber&&=\sqrt{\frac{b}{2\pi i}}\int_{\mathbb R}e^{i(at^2+bt\xi+c\xi^2+dt+e\xi)}f(t)\overline{\psi^\mu_{\beta,\alpha}}dt\\\
\nonumber&&=\sqrt{\frac{b}{2\pi i}}\int_{\mathbb R}e^{i(at^2+bt\xi+c\xi^2+dt+e\xi)+ia(t^2-\beta^2)+id(t-\beta)}f(t)\overline{\psi_\alpha(t-\beta)}dt\\\
\nonumber&&=e^{i(c\xi^2+e\xi-a\beta^2-d\beta)}\sqrt{\frac{b}{2\pi i}}\int_{\mathbb R}e^{i(2at^2+bt\xi+2dt)}f(t)\overline{\psi_\alpha(t-\beta)}dt\\\
\nonumber&&=e^{i(c\xi^2+e\xi-a\beta^2-d\beta)}\sqrt{\frac{b}{2\pi i}}\int_{\mathbb R}e^{i(2at^2+2dt)}e^{ib\xi t}f(t)\overline{\psi_\alpha(t-\beta)}dt\\\
\nonumber&&=e^{i(c\xi^2+e\xi-a\beta^2-d\beta)}\sqrt{\frac{b}{ i}}\mathcal F[g](b\xi)\\\
\end{eqnarray}
where \begin{equation}\label{g}
g(t)= e^{i(2at^2+2dt)}f(t)\psi_\alpha(t-\beta).\end{equation}
\end{lemma}

\begin{theorem}[Logarithmic uncertainty principle] Let $\psi\in L^2(\mathbb R)$ and $W^\mu_f(\xi,\beta,\alpha)$ be the QP-WPT of $f \in \mathcal S(\mathbb R)$ [Schwartz space]. Then, the following logarithmic inequality holds:
\begin{eqnarray*}\label{log eqn}
\nonumber \|\psi\|^2\int_{\mathbb R}\ln|t||f(t)|^2dt+\int_{\mathbb R}\int_{\mathbb R}\ln|\xi|\left|W^\mu_f(\xi,\beta,\alpha)\right|^2d\xi d\beta&&\\\
\nonumber \ge \left[\frac{\Gamma'(1/4)}{\Gamma(1/4)}-\ln\pi-\ln|b|\right]\|f\|^2\|\psi\|^2\\\
\end{eqnarray*}
\end{theorem}

\begin{proof}
For any $f\in \mathcal S(\mathbb R)$(Schwartz space in $L^2(\mathbb R),$) the logarithmic uncertainty principle for the classical Fourier transform reads\cite{fs2}
\begin{equation}\label{log1}
\int_{\mathbb R}\ln|t||f(t)|^2dt+\int_{\mathbb R}\ln|\xi||\mathcal F[f](\xi)|^2d\xi\ge\left[\frac{\Gamma'(1/4)}{\Gamma(1/4)}-\ln\pi\right]\int_{\mathbb R}|f(t)|^2dt.
\end{equation}
As  $f\in \mathcal S(\mathbb R),$ then it is evident that function $g$ given in (\ref{g}) belongs to the Schwartz space $\mathcal S(\mathbb R).$  Therefore we can replace $f$ in (\ref{log1}) by $g$ as:
\begin{equation}\label{log2}
\int_{\mathbb R}\ln|t||g(t)|^2dt+\int_{\mathbb R}\ln|\xi||\mathcal F[g](\xi)|^2d\xi\ge\left[\frac{\Gamma'(1/4)}{\Gamma(1/4)}-\ln\pi\right]\int_{\mathbb R}|g(t)|^2dt.
\end{equation}
Changing $\xi$ by $b\xi$, we obtain from (\ref{log2})
\begin{equation}\label{log3}
\int_{\mathbb R}\ln|t||g(t)|^2dt+b\int_{\mathbb R}\ln|b\xi||\mathcal F[g](b\xi)|^2d\xi\ge\left[\frac{\Gamma'(1/4)}{\Gamma(1/4)}-\ln\pi\right]\int_{\mathbb R}|g(t)|^2dt.
\end{equation}

Applying Lemma \ref{rel ft} and (\ref{g}) to (\ref{log3}), we obtain
\begin{eqnarray}\label{log4}
\nonumber\int_{\mathbb R}\ln|t||f(t)\psi_\alpha(t-\beta)|^2dt+b\int_{\mathbb R}(\ln|b|+\ln|\xi|)\left|\sqrt{\frac{i}{b}}e^{-i(c\xi^2+e\xi-a\beta^2-d\beta)}W^\mu_f(\xi,\beta,\alpha)\right|^2d\xi&&\\\
\nonumber\ge\left[\frac{\Gamma'(1/4)}{\Gamma(1/4)}-\ln\pi\right]\int_{\mathbb R}|f(t)\psi_\alpha(t-\beta)|^2dt.\\\
\end{eqnarray}
On further simplifying (\ref{log4}), we get
\begin{eqnarray}\label{log5}
\nonumber\int_{\mathbb R}\ln|t||f(t)\psi_\alpha(t-\beta)|^2dt+\int_{\mathbb R}\ln|b|\left|W^\mu_f(\xi,\beta,\alpha)\right|^2d\xi+\int_{\mathbb R}\ln|\xi|\left|W^\mu_f(\xi,\beta,\alpha)\right|^2d\xi&&\\\
\nonumber\ge\left[\frac{\Gamma'(1/4)}{\Gamma(1/4)}-\ln\pi\right]\int_{\mathbb R}|f(t)\psi_\alpha(t-\beta)|^2dt.\\\
\end{eqnarray}
On integrating both sides of (\ref{log5}) with respect to $\beta$ , we have
\begin{eqnarray}\label{log6}
\nonumber\int_{\mathbb R}\int_{\mathbb R}\ln|t||f(t)\psi_\alpha(t-\beta)|^2dtd\beta+\ln|b|\int_{\mathbb R}\int_{\mathbb R}\left|W^\mu_f(\xi,\beta,\alpha)\right|^2d\xi d\beta\\\
\nonumber+\int_{\mathbb R}\int_{\mathbb R}\ln|\xi|\left|W^\mu_f(\xi,\beta,\alpha)\right|^2d\xi d\beta&&\\\
\nonumber\ge\left[\frac{\Gamma'(1/4)}{\Gamma(1/4)}-\ln\pi\right]\int_{\mathbb R}\int_{\mathbb R}|f(t)\psi_\alpha(t-\beta)|^2dtd\beta.\\\
\end{eqnarray}
Now  using (\ref{con 2}) in (\ref{log6}), we get
\begin{eqnarray*}\label{log7}
\nonumber \|\psi\|^2\int_{\mathbb R}\ln|t||f(t)|^2dt+\int_{\mathbb R}\int_{\mathbb R}\ln|\xi|\left|W^\mu_f(\xi,\beta,\alpha)\right|^2d\xi d\beta&&\\\
\nonumber\ge\left[\frac{\Gamma'(1/4)}{\Gamma(1/4)}-\ln\pi\right]\|f\|^2\|\psi\|^2   -\ln|b|\|f\|^2\|\psi\|^2\\\
=\left[\frac{\Gamma'(1/4)}{\Gamma(1/4)}-\ln\pi-\ln|b|\right]\|f\|^2\|\psi\|^2
\end{eqnarray*}
which completes the proof.
\end{proof}

\begin{theorem}\label{hsb thm}
For $\psi,f\in L^2(\mathbb R)$ and $W^\mu_f(\xi,\alpha,\beta)$ be the QP-WPT of the signal $f,$ then the following inequality holds:
\begin{eqnarray}\label{hsb eqn}
\int_{\mathbb R}t^2|f(t)|^2dt\int_{\mathbb R}\int_{\mathbb R}\xi^2|W^\mu_f(\xi,\beta,\alpha)|^2d\xi d\beta \ge\left(\frac{1}{2|b|}\|f\|^2\|\psi\|\right)^2.
\end{eqnarray}

\end{theorem}
\begin{proof}

The classical Heisenberg-Pauli-Weyl inequality in the QPFT domain(see \cite{fs2} Theorem 3.2 ) is given by
\begin{equation}\label{hsb1}
\int_{\mathbb R}t^2|f(t)|^2dt\int_{\mathbb R}\xi^2|\mathcal Q_\mu[f](\xi)|^2d\xi\ge\left(\frac{1}{2|b|}\int_{\mathbb R}|f(t)|^2dt\right)^2.
\end{equation}
Using the inverse transform for the QPFT into the LHS  and Plancherel identity for QPFT (\ref{}) into the RHS of the  (\ref{hsb1}),  we have
\begin{equation}\label{hsb2}
\int_{\mathbb R}t^2|Q^{-1}_\mu [Q_\mu [f](\xi)]|^2(t)dt\int_{\mathbb R}\xi^2|\mathcal Q_\mu[f](\xi)|^2d\xi\ge\left(\frac{1}{2|b|}\int_{\mathbb R}|Q_\mu [f](\xi)|^2d\xi\right)^2.
\end{equation}

For $f,\mathcal Q^\mu[f]\in L^2(\mathbb R)$  we have $W^\mu_f(\xi,\beta,\alpha)\in L^2(\mathbb R),$ so replacing $\mathcal Q^\mu[f]$ by $W^\mu_f(\xi,\beta,\alpha)$ in (\ref{hsb2}), we have
\begin{equation}\label{hsb3}
\int_{\mathbb R}t^2|Q^{-1}_\mu [W^\mu_f(\xi,\beta,\alpha)]|^2dt\int_{\mathbb R}\xi^2|W^\mu_f(\xi,\beta,\alpha)|^2d\xi\ge\left(\frac{1}{2|b|}\int_{\mathbb R}|W^\mu_f(\xi,\beta,\alpha)|^2d\xi\right)^2.
\end{equation}
Which implies
\begin{equation}\label{hsb4}
\left(\int_{\mathbb R}t^2|Q^{-1}_\mu [W^\mu_f(\xi,\beta,\alpha)]|^2dt\right)^{1/2}\left(\int_{\mathbb R}\xi^2|W^\mu_f(\xi,\beta,\alpha)|^2d\xi\right)^{1/2}\ge\frac{1}{2|b|}\int_{\mathbb R}|W^\mu_f(\xi,\beta,\alpha)|^2d\xi.
\end{equation}
Now integrating (\ref{hsb4}) both sides by $\beta$, we have
\begin{eqnarray}\label{hsb4}
\nonumber\int_{\mathbb R}\left(\int_{\mathbb R}t^2|Q^{-1}_\mu [W^\mu_f(\xi,\beta,\alpha)]|^2dt\right)^{1/2}\left(\int_{\mathbb R}\xi^2|W^\mu_f(\xi,\beta,\alpha)|^2d\xi\right)^{1/2}d\beta&&\\\
\nonumber\ge\frac{1}{2|b|}\int_{\mathbb R}\int_{\mathbb R}|W^\mu_f(\xi,\beta,\alpha)|^2d\xi d\beta&&\\\
\end{eqnarray}
Now applying Cauchy-Schwartz inequality, (\ref{hsb4}) yields
\begin{eqnarray}\label{hsb5}
\nonumber\left(\int_{\mathbb R}\int_{\mathbb R}t^2|Q^{-1}_\mu [W^\mu_f(\xi,\beta,\alpha)]|^2dtd\beta\right)^{1/2}\left(\int_{\mathbb R}\int_{\mathbb R}\xi^2|W^\mu_f(\xi,\beta,\alpha)|^2d\xi d\beta\right)^{1/2}&&\\\
\nonumber\ge\frac{1}{2|b|}\int_{\mathbb R}\int_{\mathbb R}|W^\mu_f(\xi,\beta,\alpha)|^2d\xi d\beta&&\\\
\end{eqnarray}
Now, using (\ref{con 2}) in (\ref{hsb5}), we obtain
\begin{eqnarray}\label{hsb6}
\nonumber\left(\int_{\mathbb R}\int_{\mathbb R}t^2|f(t)\psi(t-\beta)|^2dtd\beta\right)^{1/2}\left(\int_{\mathbb R}\int_{\mathbb R}\xi^2|W^\mu_f(\xi,\beta,\alpha)|^2d\xi d\beta\right)^{1/2}&&\\\
\nonumber\ge\frac{1}{2|b|}\|f\|^2\|\psi\|^2.&&\\\
\end{eqnarray}
On further simplifying (\ref{hsb6}), we get
\begin{eqnarray}\label{hsb7}
\nonumber\left(\int_{\mathbb R}t^2|f(t)|^2dt\right)^{1/2}\left(\int_{\mathbb R}\int_{\mathbb R}\xi^2|W^\mu_f(\xi,\beta,\alpha)|^2d\xi d\beta\right)^{1/2}&&\\\
\nonumber\ge\frac{1}{2|b|}\|f\|^2\|\psi\|.&&\\\
\end{eqnarray}
Which implies
\begin{eqnarray*}
\int_{\mathbb R}t^2|f(t)|^2dt\int_{\mathbb R}\int_{\mathbb R}\xi^2|W^\mu_f(\xi,\beta,\alpha)|^2d\xi d\beta \ge\left(\frac{1}{2|b|}\|f\|^2\|\psi\|\right)^2.
\end{eqnarray*}
Which completes the proof.
\end{proof}
\begin{remark}\label{remm}By varying the parameter $\mu=(a,b,c,d,e)$
the Heisenberg-type inequality (\ref{hsb eqn}), embodies certain existing Heisenberg-type inequalities and also give birth to some novel Heisenberg-type inequalities which are yet to be reported in the open literature which are listed below:

\begin{itemize}
\item For $\mu=(a/2b,-1/b,c/2b,0,0)$, the Heisenberg-type inequality (\ref{hsb eqn}) boils down to the novel Heisenberg inequality for linear canonical wave packet
transform(see Theorem 6.2 \cite{x36})\begin{equation*}
\int_{\mathbb R}t^2|f(t)|^2dt\int_{\mathbb R}\int_{\mathbb R}\xi^2|W^\mu_f(\xi,\beta,\alpha)|^2d\xi d\beta \ge\left(\frac{|b|}{2}\|f\|^2\|\psi\|\right)^2.
\end{equation*}
\item For $\mu=(\cot\theta,-\csc\theta,\cot\theta,0,0)$, $\theta\ne n\pi$,
 we can obtain the novel Heisenberg inequality for the fractional wave packet
transform\begin{equation*}
\int_{\mathbb R}t^2|f(t)|^2dt\int_{\mathbb R}\int_{\mathbb R}\xi^2|W^\mu_f(\xi,\beta,\alpha)|^2d\xi d\beta \ge\left(\frac{\sin\theta}{2}\|f\|^2\|\psi\|\right)^2.
\end{equation*}

\item For $\mu=(0,-1,1,0,0)$,
we can obtain the novel Heisenberg inequality for the  the classical
wave packet
transform\begin{equation*}
\int_{\mathbb R}t^2|f(t)|^2dt\int_{\mathbb R}\int_{\mathbb R}\xi^2|W^\mu_f(\xi,\beta,\alpha)|^2d\xi d\beta \ge\left(\frac{1}{2}\|f\|^2\|\psi\|\right)^2.
\end{equation*}
\end{itemize}

\end{remark}

\section{Conclusion} \label{sec 5}
Based on quadratic phase Fourier transform (QPFT) and the classical wave packet transform (WPT) theory, we in this paper propose a novel integral transform coined as quadratic phase wave packet transform (QP-WPT) which rectifies the
limitations of the WPT and QPFT. Overall, it not only combines the advantages of QPFT and WPT, but also preserves the properties of its conventional counterpart, and has better mathematical properties. Besides studying some notable inequalities 
and  the fundamental properties including the Moyal's formula, inversion formula and a
reproducing kernel, we also formulated several classes of uncertainty inequalities, such
as  Leib's uncertainty principle, the logarithmic uncertainty inequality and the Heisenberg inequality.

\end{document}